\documentclass[letterpaper]{article} 
\usepackage{aaai2026}  
\usepackage{times}  
\usepackage{helvet}  
\usepackage{courier}  
\usepackage[hyphens]{url}  
\usepackage{graphicx} 
\usepackage{natbib}  
\usepackage{caption} 
\usepackage{algorithm}
\usepackage{algorithmic}
\usepackage{booktabs}
\usepackage{amsmath}
\usepackage{amsfonts}
\usepackage{makecell}
\usepackage{amssymb} 

\newtheorem{theorem}{Theorem}

\usepackage{newfloat}
\usepackage{listings}
\DeclareCaptionStyle{ruled}{labelfont=normalfont,labelsep=colon,strut=off} 
\floatstyle{ruled}
\newfloat{listing}{tb}{lst}{}
\floatname{listing}{Listing}
\title{Refined Thompson Learning for Adaptive Bandits:
Power-Efficient Flexibility Scheduling Across Data Centers}
\author{
Yifu Ding$^{1,2}$\equalcontrib,
Zixi Chen$^{3}$\equalcontrib,
Ruicheng Ao$^{4}$,
David Simchi-Levi$^{4}$,
Thomas Magnanti$^{2, 5}$
}
\affiliations{
$^{1}$MIT Energy Initiative, Massachusetts Institute of Technology \\
$^{2}$Sloan School of Management, Massachusetts Institute of Technology\\
$^{3}$School of Mathematical Sciences, Peking University\\
$^{4}$Institute for Data, Systems, and Society, Massachusetts Institute of Technology \\
$^{5}$Electrical Engineering \& Computer Science, Massachusetts Institute of Technology \\
}

\usepackage{bibentry}

\begin{document}

\maketitle

\begin{abstract}
The rapid growth of large-scale AI workloads in data centers has placed increasing pressure on power grids in recent years. Since power systems must continuously balance supply and demand, there is growing interest in leveraging data-center workload flexibility as a grid service. We propose a contextual restless multi-armed bandit (CRMAB) framework in which a grid operator requests load reductions without observing internal job-scheduling decisions. Under an index-ability guarantee, each data center or physical machine is modeled as a Markov decision process (MDP) over a cyclic virtual-machine (VM) job queue, with unknown rewards and transition dynamics that are learned online using Thompson sampling and Whittle-index policies. To improve learning under sparse and noisy observations, the framework augments an adaptive Thompson--Whittle (TW) policy with domain-informed transition priors and gated prior mixing. In baseline experiments, the best refined adaptive variant achieves 91.4\% of the oracle reward after 100 rounds and 96.8\% after 1,000 rounds. Across a 16-setting stress test spanning different state-space sizes and levels of contextual feature noise, the best refined variant consistently outperforms the original TW policy with high confidence while remaining competitive with EXP4. A graph-based prior further incorporates data-center hardware constraints, including operating frequency and computing-resource limits. Overall, our results demonstrate the economic potential of data-center flexibility as a grid service and highlight the importance of high-quality, open-source AI workload traces for developing and evaluating such services.
\end{abstract}

\section{Introduction}

The rapid growth of large-scale AI workloads has significantly increased energy consumption in data centers, placing higher greenhouse gas (GHG) emissions and grid stress worldwide \citep{kaack_aligning_2022, sajadi_power_2026}. In the U.S., electricity demand from data centers is projected to increase by 7\%--12\% by 2030~\citep{epri_powering_2024}, while globally, data center electricity consumption is expected to at least double over the next five years~\citep{chen_data_2025}. Flexible computing can mitigate these impacts by shifting or reordering workloads in response to grid or carbon conditions \citep{senga_flexible_2026, cheng_ai_2026}. Major cloud providers have begun voluntarily participating in demand response programs. For example, Google has collaborated with electric utilities on demand response initiatives to alleviate grid stress ~\citep{terrall_how_2025}. Recent demonstrations using workload power-orchestration software show that job rescheduling within a data center can directly respond to grid power-reduction signals \citep{colangelo_ai_2025}.

Existing approaches include workload scheduling \citep{dai_throughput-optimal_nodate, liu_watts_2026}, dynamic voltage frequency scaling (DVFS) \citep{gu_greenflow_2025, ding_distributionally_2025, crozier_potential_2025}, and energy-aware AI systems \citep{song_energylens_2026,chung_openg2g_nodate}. However, detailed workload information and internal scheduling policies may be unavailable to grid operators because of privacy constraints. Despite these advances, data center operators are often reluctant to disclose detailed workload characteristics or internal scheduling policies because of privacy, security, and commercial concerns. In practice, they are typically willing to share only coarse-grained information, such as workload flexibility classes, aggregate power consumption, or historical VM statistics. Although formal grid codes governing data center demand response are not yet in place, future grid operators are expected to coordinate workload flexibility requests across geographically distributed data centers under highly uncertain operating conditions, requiring the balance of economic benefits, grid reliability, and quality-of-service guarantees.

\section{Related Works}

    \paragraph{Job scheduling in dynamic data center management:} Open-sourced VM traces provide useful aggregate computing, memory, and workload information \citep{cortez_resource_2017, zhang_ivmr_2025}, motivating learning-based coordination from limited contextual observations.  \citep{guillen-perez2026dcclusteropt} introduced a benchmark reinforcement learning (RL) framework for scheduling geo-distributed data centers, leveraging VM datasets that capture AI workloads alongside diverse environmental and grid-related data, such as electricity prices, carbon intensity, and weather conditions. \citep{zhang_ivmr_2025} collects industrial-scale VM datasets and develops a comprehensive job rescheduling framework using optimization, metaheuristic, heuristic, and ML-based methodologies. These works, including \citep{zhang_ivmr_2025} and \citep{guillen-perez2026dcclusteropt}, generally assume full observability of all information contained in the VM traces.

    \paragraph{RMAB and index-based policy:} The CRMAB framework is widely employed for optimal resource allocation under uncertainty, where each arm’s state space is modeled as a MDP. However, CRMAB problems are generally intractable and are known to be NP-hard with the expanding state space. To address this challenge, index policies are introduced to decouple the arms and enable tractable, per-arm decision making that leads to near-optimal solutions. \citep{whittle_restless_1988} proposed an index policy based on Lagrangian relaxation, which evaluates the marginal value of activating each arm independently. The Whittle index has been applied to stochastic deadline scheduling problems \citep{yu_deadline_2016}, \citep{graczova_generalized_2014}, where indexability is established and efficient algorithms are derived to compute the index in closed form. Arms are then prioritized and activated according to their index values. The index-based RMAB has also been applied to model demand responses in electric grids  ~\citep{chen_contextual_2024} with the perfect contextual information.

\paragraph{Refined Thompson Learning for bandits:} During the index policy computation in CRMAB, the reward and transition functions are typically unknown. Thompson learning is a commonly used approach in bandit settings to infer transition dynamics and reward functions from observed data.  The parameterization and updates of the prior distribution depend on state visits, which are observed adaptively over the course of bandit interactions. The choice of the prior distribution can significantly affect learning performances, and a poorly-chosen prior may lead to slower convergence and lower accumulated rewards \citep{russo_tutorial_2018}. A common baseline for the prior distribution is to assume specified continuous or discrete distributions, such as Gaussian distributions or Gaussian mixture models (GMMs) \citep{liang_context_2025}.

\subsection{Contributions}

We thus develop \textbf{\underline{R}}obust and \textbf{\underline{A}}daptive \textbf{\underline{C}}omputing and \textbf{\underline{E}}nergy \textbf{\underline{R}}esource coordination plus the refined strategies (\textbf{RACER+}), as in Figure \ref{framework}. Building on our previous work \cite{chen_robust_2026}, this paper makes three main contributions: (i) a contextual restless multi-armed bandit (CRMAB) formulation for data-center flexibility through batch-level VM rescheduling; (ii) an adaptive Thompson–Whittle (TW) policy that gradually shifts from exploration to Whittle-index-guided decision making as state observations accumulate; and (iii) domain-informed transition refinements, including structural support and gated priors, to improve learning under sparse and noisy observations.

\section{Methodology}
\label{tab:methods}

\subsection{Model formulation}

Consider $\mathcal N$ data centers, each executing batches of $N_j$ jobs from a cyclic VM queue. Without a flexibility request, the next $N_j$ jobs are executed. When activated, a data center may look ahead over $N_f\geq N_j$ jobs and select a lower-power feasible batch, with delay penalties for skipped latency-sensitive jobs. For arm $i$, state $s_{i,t}\in\mathcal S_i$ identifies the current queue/batch condition of workloads, $a_{i,t}\in\{0,1\}$ denotes passive/activated operation of demand responses, $R_i(a,s)$ is the net flexibility reward, and $P_i^a$ is the action-dependent transition kernel. Both reward and transition functions are unknown. The definition of MDP is presented in Table \ref{MDP_formulation}.

 \begin{table}[h]
\centering
\begin{tabular}{p{0.18\linewidth} p{0.75\linewidth}}
\toprule
Components & Descriptions \\
\midrule
States ($s_t \in \mathcal{S}$) 
& Position $s_i$ in the circular job queue, representing the current batch of $N_j$ jobs. \\

Actions ($a_t \in \mathcal{A}$) 
& When \(a_i = 0\), data center \(i\) executes the default next \(N_j\) jobs in the queue. The data center is activated if \(a_i = 1\) or passive if \(a_i = 0\). \\

Rewards ($r_t \in R$) 
& When action \(a_i = 0\), the reward is always zero. When action \(a_i = 1\), the reward is defined as the net benefit from rescheduling, which is computed as the power savings minus the delay penalty associated with any interactive jobs skipped during reordering. This is detailed in Supplementary Material 1 \textit{Reward functions}.\\
Transition ($P$) 
& State transitions depend on whether the arm is activated. 
$P_i^1(\cdot \mid s_{i,t})$ denotes the transition probability under activation, while $P_i^0(\cdot \mid s_{i,t})$ denotes the transition probability when not activated. \\
\bottomrule
\end{tabular}
\caption{MDP formulation for VM (re)scheduling in data centers}
\label{MDP_formulation}
\end{table}

At round $t$, the grid may activate at most $N_t$ data centers:


\begin{equation}
\max_{\pi}\;
\mathbb E_{\pi}\!\left[
\sum_{t=0}^{\infty}\beta^t
\sum_{i=1}^{\mathcal N} r_{i,t}a_{i,t}
\right],
\qquad
\sum_{i=1}^{\mathcal N}a_{i,t}\leq N_t .
\label{objective3}
\end{equation}

The framework learns each arm's reward function and action-dependent transition dynamics using Gaussian and Dirichlet--Categorical posteriors, respectively, and computes Whittle indices from the sampled models. We consider both one-dimensional states and two-dimensional product states comprising job-queue states and operation types, as detailed in a later section. The formal definition of the Whittle index and the verification of indexability are provided in \textit{Reward functions and index-ability}.

\begin{figure*}[h]
\centering
     \includegraphics[width=\linewidth]{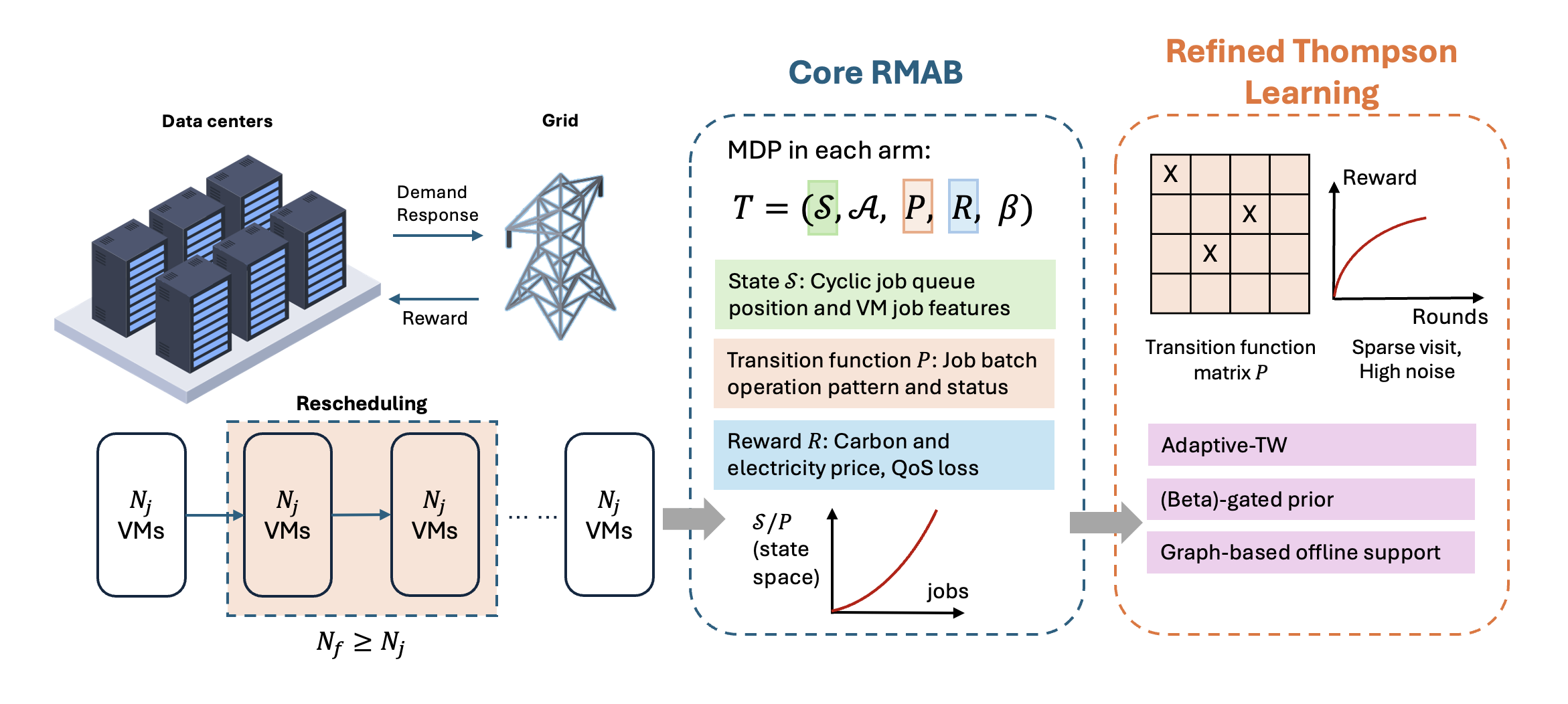}
     \caption{RACER+ framework consists of three parts: (a) Batched VM job creation, (2) Core CRMAB framework, and (3) Refined Thompson learning strategies.}
     \label{framework}
\end{figure*}

\subsection{Core Strategies for Arm Selection}

\paragraph{Contextual bandit and Whittle index} As a model-free baseline, we use a contextual bandit that conditions only on current batch features and ignores state transitions. For each arm $i$, the context $x_i(t)\in\mathbb{R}^n$ includes features such as computing/memory usage and power consumption. We assume a linear reward model,
\begin{equation}
r_{i,t}=x_i(t)^\top\theta_i+\varepsilon_{i,t},
\qquad
\varepsilon_{i,t}\sim\mathcal N(0,\sigma^2),
\end{equation}
and maintain a Bayesian posterior over $\theta_i$. At each round, we sample $\tilde\theta_i(t)$ and compute
\begin{equation}
\tilde\mu_i(t)=x_i(t)^\top\tilde\theta_i(t).
\end{equation}
The policy activates the $N_t$ arms with the highest sampled scores, whereas the Whittle-index policy selects arms according to their state-dependent indices.

\paragraph{Thompson--Whittle (TW)} 
For each arm $i$ at time $t$, the state
$s_i(t):=\psi(x_i(t))\in\mathcal S_i$ is obtained from its context.
The unknown model $\theta_i=(r_i,P_i^1,P_i^0)$ is learned using
Gaussian reward and Dirichlet--Categorical transition posteriors.
At each step, the algorithm samples a model from the posterior and
solves the corresponding subsidy MDP to compute the Whittle index.
Weakly informative priors improve learning in sparsely visited states while remaining adaptive to new observations.

\begin{align}
\tilde V_{i,\lambda}^{(t)}(s)
= \max \Big\{
&\tilde r_i(s)
+ \beta\sum_{s'} \tilde P_i^1(s,s')
  \tilde V_{i,\lambda}^{(t)}(s'), \nonumber\\
&\lambda
+ \beta \sum_{s'} \tilde P_i^0(s,s')
  \tilde V_{i,\lambda}^{(t)}(s')
\Big\}.
\label{MDP_subsidy}
\end{align}

Then, the Whittle index is computed as,

\begin{align}
W_i(s) = \inf \Bigg\{ \lambda :\;
&\lambda + \beta \sum_{s'} P_i^0(s,s') V_{i,\lambda}(s')\ge \nonumber \\ & r_i(s) + \beta \sum_{s'} P_i^1(s,s') V_{i,\lambda}(s')
\Bigg\} \label{eq:subsidy_mdp_compute}
\end{align}

\subsubsection{Refined algorithmic variants}

The aforementioned TW algorithm exhibits limitations under stress conditions, particularly when the number of learning rounds is limited or during early stages, due to sparse observations in visited states and high uncertainty in unvisited ones. To address these issues, we develop the following refined algorithmic variants.

\paragraph{Adaptive-TW} \textit{Adaptive-TW} combines a greedy exploration strategy, such as UCB, with the Whittle-index policy to ensure exploration-exploitation balance through the learning rounds. Let \(\widetilde W_i(s)\) denote \textit{Thompson-sampled} Whittle score for arm \(i\) in observed state \(s\), and let \(U_i(s,t)\) denote the UCB reward score computed from the state-level reward posterior. The refined mixer uses a visit-dependent trust weight, clipped between the upper $\tau_{\max}$ and lower trust floor $\tau_{\min}$ is given by,

\begin{equation}
\tau_i(s,t)
=
\mathrm{clip}\!\left(
\frac{O_i(s,t)}
{O_i(s,t)+2+\sqrt{|\mathcal{S}|}},
\tau_{\min},\tau_{\max}
\right)
\label{eq:adaptive_trust}
\end{equation}

where \(O_i(s,t)\) is the number of observed transitions from state \(s\) of arm \(i\), and the constant-term \(2 + \sqrt{|\mathcal{S}|}\) encodes the size of state space. In high-noise state settings, calibration may require a larger floor value to ensure reliable trust in the Whittle index. The adaptive score mixing the sampled Whittle index \(\widetilde W_i(s, t)\) and UCB strategy \(U_i(s,t)\) is
\begin{equation}
A_i(s,t)
=
\tau_i(s,t)\,z\!\left(\widetilde W_i(s,t)\right)
+
\bigl[1-\tau_i(s,t)\bigr]\,z\!\left(U_i(s,t)\right)
\label{eq:adaptive_mixing}
\end{equation}
where the function \(z(\cdot)\) standardizes scores across arms before ranking. Thus, the policy is local-first in weakly visited states and Whittle-first once transition evidence accumulates.

\subsection{Domain-knowledge-enriched refined strategies}

We propose a set of domain-knowledge-enriched refined strategies which leverage historical samples or computing resource constraints.

\paragraph{Domain-informed transition refinement.} Job queues provide structural information while preserving scheduling privacy. We blend the posterior transition sample $\widehat P_i^a$ with a domain prior $P_{0,i}^a$, optionally masking infeasible transitions:
\begin{equation}
\widetilde P_i^a(s,\cdot)=
(1-\gamma_i^a)\widehat P_i^a(s,\cdot)
+\gamma_i^a P_{0,i}^a(s,\cdot),
\label{eq:mixed_transition}
\end{equation}
where $\gamma_i^a$ decays with observations. A beta-gated variant randomizes $\gamma_i^a$ to reduce conservativeness under noisy contexts. Full prior construction and beta-gate details are in \textit{Refined Thompson-Learning Strategies}.

\section{Case Study}
\label{Case_Study}

\subsection{Experiment settings}
\label{yifu:numerical_evaluations}
We conduct three types of experiments, with configurations summarized in Table \ref{yifu:tab:settings}. 
In the baseline setting which simulates the normal grid conditions, the state space is simple and the number of rounds is large. In the transition-stress setting, the state space increases from 20 to 100 states, while the number of rounds is reduced to 200, allowing us to evaluate the refinement techniques for learning sparse transitions under limited online observations. In real-world simulations, the state space is further expanded to two-dimensional state space. For each experiment type, we run 10 random seeds on VM job sampling. We also perform the hyperparameter sweeps to identify optimal values, such as the trust floor $\tau_{min} \in (0, 0.1]$ and beta gate $G^a_i$.

\begin{table}[!h]
\centering
\resizebox{\columnwidth}{!}{
\begin{tabular}{lccc}
\toprule
& Baseline & \makecell{Transition \\ stress} & \makecell{Real-world \\simulations} \\
\midrule
Arms        & 5 & 5 & 8 \\
Budgets     & 1 & 1 & 3 \\
States      & 8 & 8--100 & 80 \\
Batch size  & 5 & 5 & 5 \\
Rounds      & 1000 & 200 & 250 \\
Features    & (a)--(d) & (a)--(d) &
\makecell[c]{(a)--(d), and \\ operational types} \\
Data sets &
\makecell{Microsoft Azure \\ VM Dataset} &
\makecell{Microsoft Azure \\ VM Dataset} &
\makecell{MIT Super-cloud \\ Datasets} \\
\bottomrule
\end{tabular}%
}
\caption{Experimental settings used in refinement ablation.}
\label{yifu:tab:settings}
\end{table}

We use Microsoft Azure VM dataset reported in ~\citep{cortez_resource_2017} and the calibrated MIT SuperCloud VM dataset ~\citep{the_mit_supercloud_mit_2021} as representative datasets for three experiments described above. Details on dataset processing and features construction are provided in \textit{Experimental Details}. Reproducible code is in Supplementary Materials, \textit{Reproducibility}.

\section{Numerical Results}
\label{tab:results}
\subsection{Refinement ablation experiments}

\paragraph{Algorithm comparison} 

We compare the refined TW variants with the original TW approach. Performance is also benchmarked against an Oracle Whittle policy as an upper bound, which assumes access to \textit{true} states with the contextual information and simply performs Whittle-index lookup.

\begin{itemize}
    \item \textbf{State Thompson (ST):}
    a model-free contextual baseline that directly learns
    state--reward values from observations.

    \item \textbf{Thompson--Whittle (TW):}
    learns reward and transition models online through posterior
    sampling and computes Whittle indices from the sampled models.

    \item \textbf{Local UCB + TW and Global UCB + TW:}
    mixed policies that combine local or global UCB scores with
    Thompson--Whittle scores.

    \item \textbf{EXP4:}
    treats Global UCB, Local UCB, and TW as experts and learns
    an adaptive mixture over their recommendations.

    \item \textbf{Oracle Whittle:}
    assumes access to the true transition and reward models and
    therefore provides an upper benchmark.
\end{itemize}

\begin{figure*}[!h]
    \centering
    \includegraphics[width=\linewidth]{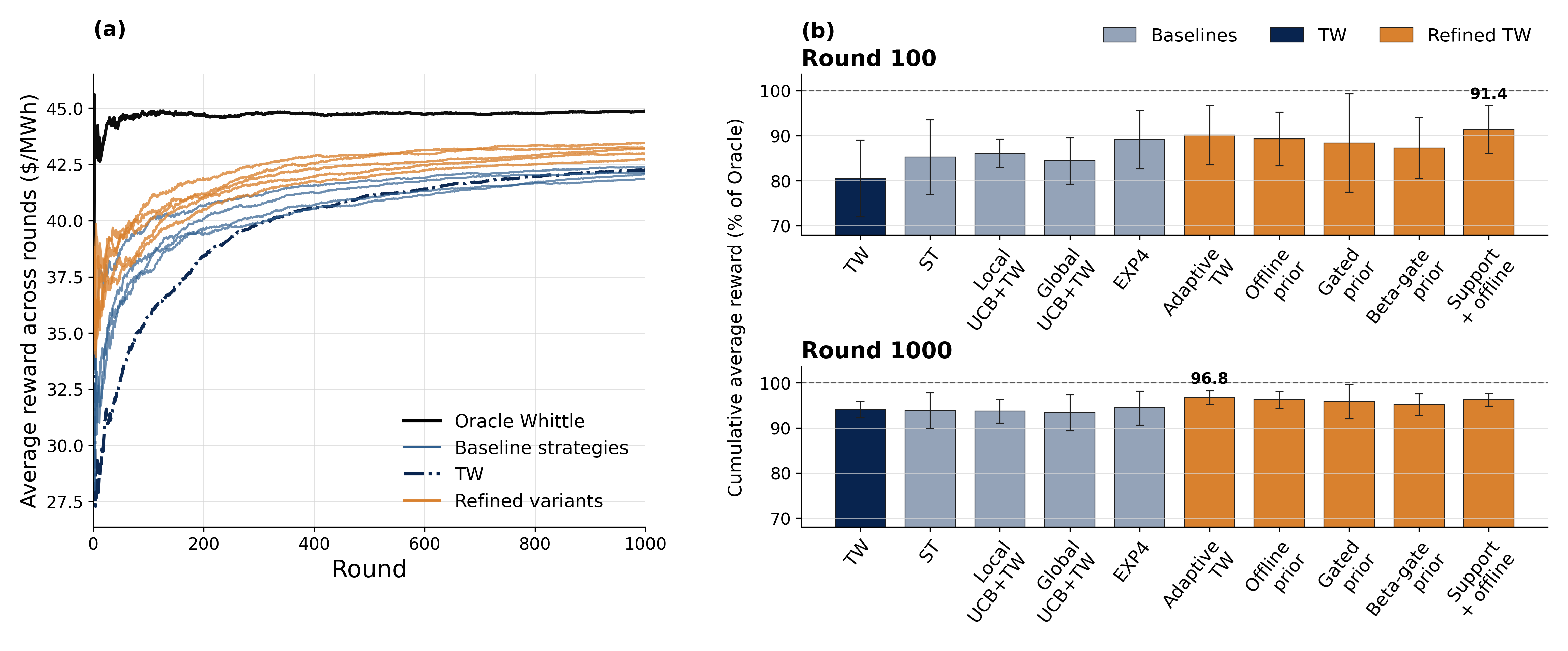}
    \caption{Baseline learning comparison (a) comparison between the refined variants and baseline arm-selection strategies over 1{,}000 rounds, and (b) their performance snapshots at 100 and 1{,}000 rounds. Refined TW variants are more stable early in learning and remain competitive after transition estimates become well learned. Bars report mean oracle-normalized cumulative reward over 10 seeds with seed-level standard deviations. The detailed data is presented in Appendix \textit{Baseline Experiments}.}
    \label{yifu:fig:overall_baseline_comparison_two_groups}
\end{figure*}

Figure~\ref{yifu:fig:overall_baseline_comparison_two_groups} shows (a) comparison between the refined variants and baseline arm-selection strategies over 1{,}000 rounds, and (b) their performance snapshots at 100 and 1{,}000 rounds. The refined TM variants are evaluated against five aforementioned baseline strategies. All the strategies are preformed using the optimal hyperparameter sweep. Across 1{,}000 rounds, the refined variants in orange consistently outperform the baseline strategy family. In Figure~\ref{yifu:fig:overall_baseline_comparison_two_groups}(b), the baseline methods perform worse and heterogeneously in the early learning stage, whereas the refined variants remain stable, with all variants achieving over 85\% of the oracle reward. Adaptive TW with the support or offline prior is the best strategy and attains 91.4\% of the oracle reward. In the later stage at 1,000 rounds, Adaptive TW performs best, achieving 96.8\% of the oracle reward on average.

However, strong support constraints or informative priors can become overly conservative once transitions are well learned (e.g., at 1,000 rounds), reducing reward relative to Adaptive TW. Transition \(L_1\) error and off-support leakage show that deterministic offline support performs best in sparse-data regimes, while gated priors provide softer regularization that improves with more data. See Supplementary Material, \textit{Additional Ablation Results}.

\subsection{Contextual-noise stress tests}
We next evaluate the learning problem under increasingly challenging conditions by varying state-space size
$|\mathcal S|\in\{8,20,50,100\}$ and contextual corruption
$\rho\in\{0,0.1,0.2,0.3\}$ using only 100 learning rounds. With probability $\rho$, an arm's observed state is replaced by a sampled state from Gaussian distribution as noise. This produces 16 combinations of dimensionality and observation noise. Every refined variant beats TW in all sixteen combinations by \(15\)--\(34\)\% of oracle, with the strong confidence validated as confirmed by statistical significance tests (p-values). Refined policies are also competitive with EXP4. The complete table with confidence intervals and variant in each setting are moved to \textit{Contextual-Noise Robustness Tests}.

\subsection{Empirical Evaluations on AI Workloads and Locational Electricity Prices (LMPs)}

We further evaluate the adaptive bandit on AI workloads with peak-hour LMPs (19:00) in Texas under AI hardware constraints. We simulate eight data centers (arms) at six different locations in Texas, and two data centers are selected in each round. Each arm has an extended two-dimensional state $s=(q,o)$, where $q$ is the job-queue state and $o$ is the job operation type. Theorem~\ref{Indexability_inheritance} in the Appendix, \textit{Indexability inheritance}, shows that indexability is preserved when we scale the reward by an operation-dependent factor $\epsilon_o$ tied to $o$. This factor permits a finer-grained classification of heterogeneous workloads under varying AI hardware conditions (i.e., DVFS and computing node constraints). Furthermore, we allow the data center to look ahead multiple batches when the demand response request is activated. We thus construct a graph-based prior to map different scheduling options while considering the batch-level computing resource limit (i.e., core-hours) in each arm, where infeasible rescheduling options are masked.

\begin{figure}
    \centering
    \includegraphics[width=\linewidth]{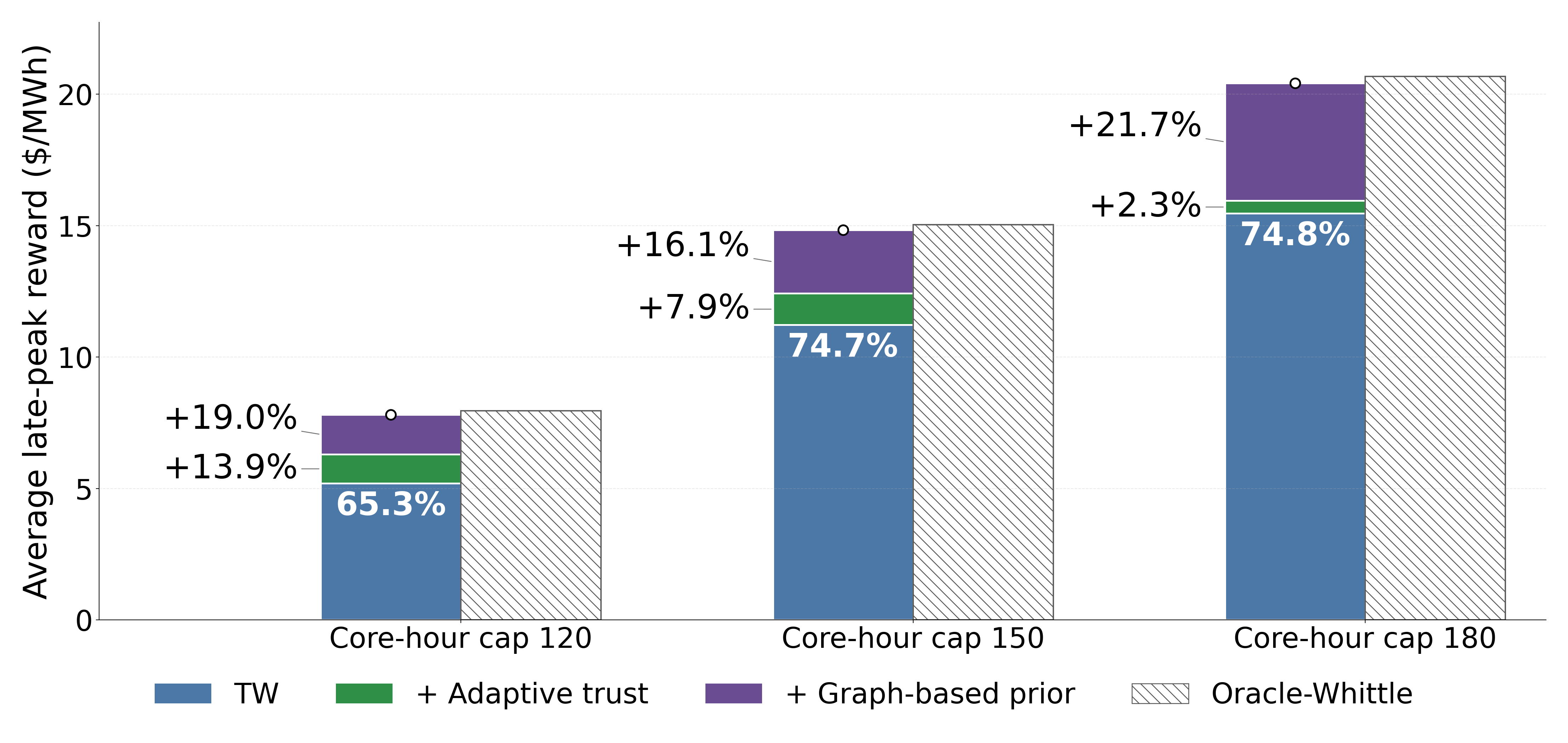}
    \caption{Reward breakdown by refinement strategy under the maximum core-hour limits of 120, 150, and 180, averaged over four random seeds. Patterned bars denote the Oracle Whittle rewards, and the percentage labels indicate each component's reward as a percentage of the Oracle.}
    \label{fig:overall_baseline_comparison}
\end{figure}

We simulate three AI hardware conditions by varying the maximum core-hour constraint. Other settings, such as DVFS, can be analyzed in this framework by applying calibrated scale factors $\epsilon_o$. Full experiments are detailed in the Appendix, \textit{Empirical Evaluation with AI Workloads and LMPs}. Figure~\ref{fig:overall_baseline_comparison} reports the average reward gains across 600 rounds under three core-hour constraints. When the core-hour constraint reduces from 180 to 120, the total reward gain reduces from \$20.4/MWh to \$7.8/MWh. The graph-based prior contributes 19.0\%-21.7\% of the reward gain, relative to the oracle benchmark across three scenarios. The adaptive mechanism becomes increasingly important as the core-hour constraint tightens, with its contribution increasing from 2.3\% to 13.9\% relative to the Oracle benchmark. Tighter resource constraints leave less room to recover from suboptimal decisions, making the balance between exploiting the learned Whittle indices and exploring uncertain arms more critical. The adaptive TW policy achieves this balance more effectively by leveraging the high-quality prior while adapting to online observations.

\section{Conclusions and Ethical Considerations}
\label{section:discussion}

We develop an adaptive bandit framework that models data-center flexibility as online learning over partially observed job-queue dynamics. Adaptive trust mixing improves early exploration, while domain-informed transition priors provide structural guidance under sparse or noisy observations. Experiments demonstrate strong early-learning performance and consistent improvements over TW across the contextual-noise stress sweep. We further show that an offline graph-based prior can encode scheduling preferences and mask infeasible paths subject to data-center resource constraints. Deployment, however, requires workload-specific calibration while respecting latency and service-level requirements. From an ethical perspective, practical deployment should preserve the confidentiality of workload traces and operational policies, since job characteristics, scheduling decisions, and infrastructure utilization may reveal commercially sensitive information. Our framework addresses this concern by learning from aggregate queue dynamics and reusable transition priors rather than exposing individual workload traces or proprietary scheduling logic. More broadly, it promotes reproducible AI workload and energy traces for research while supporting privacy-preserving evaluation of computing flexibility as a reliable grid service, advancing sustainable AI without revealing sensitive internal scheduling information.

Deployment relies on contextual measurements of hardware-level energy efficiency and system-software effects on AI workloads. Initiatives such as MLCommons \cite{tschand_mlperf_2025} have standardized such benchmarks; we hope RACER+ further encourages open, reproducible AI/ML workload datasets and privacy-preserving benchmarking practices that enable rigorous research on data-center demand flexibility without compromising confidential operational data or customer workloads.

\section{Acknowledgments}

Yifu Ding is supported by the MIT Shell Energy Scholar Program at the
MIT Energy Initiative. Zixi Chen is supported by MIT Undergraduate Research Opportunities Program (UROP). We would like to thank Tina Yuhan Ma's help in processing MIT Supercloud Datasets at MIT Energy Initiative UROP. 

\bibliography{CDC_data_center_bandit_TCCML}

\newpage

\section{Appendix}

\subsection{Reward functions and index-ability}
\label{app:method_indexability}

\subsubsection{Rescheduling and Reward Functions}
\label{app:reward}

Each arm represents a data center with a sequence of VM jobs. At state \(s\), let \(B_s\) denote the default batch and \(W_s\) a short look ahead window. Under activation, a batch \(C_s \subset W_s\) of the same size as \(B_s\) is selected to minimize power consumption. Let \(P_j\), \(Q_j\), and \(I_j\in\{0,1\}\) denote the power cost, QoS delay cost, and interaction indicator of job \(j\), respectively. The default and selected power consumptions are
\[
P_{\mathrm{def}}(s)
=
\sum_{j\in B_s} P_j,
\qquad
P_{\mathrm{chosen}}(s)
=
\sum_{j\in C_s} P_j.
\]
The delayed jobs are defined as \(D_s = B_s \cap [C_s]^c\), with QoS penalty
\[
C_{\mathrm{delay}}(s)
=
\sum_{j\in D_s} I_j Q_j.
\]
The active reward for arm \(i\) in state \(s\) is
\[
\begin{aligned}
&r_i(s)
=
\max\Big\{
0,\;
\lambda_{\mathrm{LMP}}
\big[
P_{\mathrm{def}}(s)
-
P_{\mathrm{sel}}(s)
\big] -
\lambda_{\mathrm{delay}}
c_{\mathrm{delay}}(s)
\Big\}
\end{aligned}
\]
where \(\lambda_{\mathrm{LMP}}= \$0.03/\mathrm{kWh}\) denotes the LMP and \(\lambda_{\mathrm{delay}}\) is the delay-cost multiplier. Passive arms receive zero reward, while active arms capture the electricity cost savings from VM rescheduling, but are penalized by QoS degradation due to delays in interactive jobs.

\subsubsection{Real-World Evaluation with AI Workloads and Electricity Prices}
\label{tab:carbon_and_electricity_prices}

To consider each arm represents a data center and each
state is a product state
\[
    s=(q,o),
\]

where \(q\) is the state of job queue, and \(o\) is the operation type. For example, the operation dimension $o:=0$ denotes the batch-flexible training job and $o:=1$ denotes the interactive-heavy inference jobs  

\[
\begin{aligned}
r_i(q,o)
=
&\max\Bigl\{
0,\; 
\lambda_{\mathrm{LMP}} \cdot
\epsilon_o \sum_{o \in \mathcal{O}}\left[P_{\mathrm{def}}(q, o)
- P_{\mathrm{chosen}}(q, o)\right] \\
& -\lambda_{\mathrm{delay}} (o)\,
c_{i}(q,o)
\Bigr\}
\end{aligned}
\]


\subsection{Whittle-Index Computation and Indexability}
\label{app:indexability}

\cite{whittle_restless_1988} proposed an index policy to strategically calculate the marginal reward for activating each arm independently based on Lagrangian relaxation, which is as follows, 

\textit{Definition 1 (Whittle index \cite{whittle_restless_1988})}: 
For each arm $i$ at state $s$, consider a single-arm MDP with a passivity
subsidy $\lambda$. Let $V_{i,\lambda}(s)$ be the optimal discounted value.
The Whittle index $W_i(s)$ is the smallest subsidy that makes passivity
optimal under indexability:
\begin{align}
W_i(s) = \inf \Bigg\{ \lambda :\;
&\lambda + \beta \sum_{s'} P_i^0(s,s') V_{i,\lambda}(s')\ge \nonumber\\ &r_i(s) + \beta \sum_{s'} P_i^1(s,s') V_{i,\lambda}(s')
\Bigg\} \label{eq:subsidy_mdp_compute}
\end{align}

\textit{Proof of indexability}: Under indexability, arms are ranked by $W_i(s)$ and the highest-index arms are
activated each round. Equivalently, define
\begin{align}
Q_{i,\lambda}^{1}(s)
&=
R_i(s,1)
+\beta \sum_{s'} P_i^{1}(s,s')V_{i,\lambda}(s'),
\\
Q_{i,\lambda}^{0}(s)
&=
R_i(s,0)+\lambda
+\beta \sum_{s'} P_i^{0}(s,s')V_{i,\lambda}(s').
\end{align}
For a given $\lambda$, the passive-optimal set is
\begin{equation}
\mathcal P_i(\lambda)
=
\{s\in\mathcal S_i: Q_{i,\lambda}^{0}(s)\ge Q_{i,\lambda}^{1}(s)\}.
\end{equation}

Arm $i$ is \emph{indexable} if $\mathcal P_i(\lambda)$ is nondecreasing in
$\lambda$ by set inclusion and expands from $\emptyset$ to $\mathcal S_i$ as
$\lambda$ increases from $-\infty$ to $+\infty$. Thus, for any
$\lambda_1\le \lambda_2$,
\[
\mathcal P_i(\lambda_1)\subseteq \mathcal P_i(\lambda_2).
\]
Under indexability, the Whittle index can also be written as
\begin{equation}
W_i(s)
=
\inf\{\lambda:\; Q_{i,\lambda}^{0}(s)\ge Q_{i,\lambda}^{1}(s)\}.
\label{eq:whittle_index}
\end{equation}
At each round $t$, the aggregator activates the $N_t$ arms with the largest
indices $W_i(s_{i,t})$.

\textit{Proposition 1 (Sufficient condition for indexability):}
For arm $i$, define
\begin{equation}
\Delta_{i,\lambda}(s)
=
Q_{i,\lambda}^{1}(s)-Q_{i,\lambda}^{0}(s),
\qquad s\in\mathcal S_i .
\end{equation}
If $\Delta_{i,\lambda}(s)$ is continuous and non-increasing in $\lambda$ for
every $s\in\mathcal S_i$, and there exist $\lambda_{\min}$ and
$\lambda_{\max}$ such that all states are active-optimal for
$\lambda\le \lambda_{\min}$ and passive-optimal for
$\lambda\ge \lambda_{\max}$, then arm $i$ is indexable.

\textit{Proof:}
Since $\Delta_{i,\lambda}(s)$ is continuous and non-increasing in $\lambda$,
each state switches from active-optimal to passive-optimal at most once. Hence,
\[
\mathcal P_i(\lambda)
=
\{s\in\mathcal S_i:\Delta_{i,\lambda}(s)\le 0\}
\]
expands monotonically with $\lambda$. The boundary conditions imply that
$\mathcal P_i(\lambda)$ evolves from $\emptyset$ to $\mathcal S_i$, so arm $i$
is indexable.

In our implementation, indexability is checked numerically by solving
\eqref{eq:subsidy_mdp_compute} over a subsidy grid and verifying that
$\mathcal P_i(\lambda)$ expands monotonically.

\subsection{Indexability under a Workload-Type Reward Multiplier}

\begin{theorem}[Indexability inheritance]
\label{Indexability_inheritance}
If the base arm is Whittle-indexable, then the product-state arm is
Whittle-indexable. Moreover,
\[
W(q,o)=\epsilon_o\bar W(q).
\]
\end{theorem}

Under passivity subsidy $\lambda$, the product-state Bellman equation is

\begin{align}
V_\lambda(q,o)
=
\max\Bigg\{
&\epsilon_o r(q)
+\beta\sum_{q'} P^1(q,q')V_\lambda(q',o),
\nonumber\\
&\lambda
+\beta\sum_{q'} P^0(q,q')V_\lambda(q',o)
\Bigg\}.
\end{align}

Set $\mu=\lambda/\epsilon_o$. Substituting into the base Bellman equation yields
\[
V_\lambda(q,o)
=
\epsilon_o\,\bar V_{\lambda/\epsilon_o}(q),
\]
where uniqueness follows from the contraction property of the discounted Bellman operator. Consequently,
\[
Q^a_\lambda(q,o)
=
\epsilon_o\,\bar Q^a_{\lambda/\epsilon_o}(q),
\qquad a\in\{0,1\}.
\]
Since $\epsilon_o>0$, multiplication by $\epsilon_o$ preserves the ordering of the action values. Therefore,
\[
(q,o)\in\mathcal{P}(\lambda)
\quad\Longleftrightarrow\quad
q\in\bar{\mathcal{P}}\!\left(\frac{\lambda}{\epsilon_o}\right).
\]
It follows that the passive set of the product-state arm is
\[
\mathcal{P}(\lambda)
=
\bigcup_{o\in\mathcal{O}}
\left(
\bar{\mathcal{P}}\!\left(\frac{\lambda}{\epsilon_o}\right)
\times\{o\}
\right).
\]

$\lambda_1\le\lambda_2$, then
$\lambda_1/\epsilon_o\le\lambda_2/\epsilon_o$ for every $o$.
Indexability of the base arm therefore implies
\[
\mathcal{P}(\lambda_1)\subseteq\mathcal{P}(\lambda_2).
\]
The limits $\mathcal{P}(\lambda)\to\varnothing$ as $\lambda\to-\infty$ and
$\mathcal{P}(\lambda)\to\mathcal{Q}\times\mathcal{Q}$ as $\lambda\to+\infty$ are inherited
in the same way. Thus the product-state arm is indexable. Finally,
\[
W(q,o)
=
\inf\left\{
\lambda:
q\in\bar{\mathcal{P}}\!\left(\frac{\lambda}{\epsilon_o}\right)
\right\}
=
\epsilon_o\bar W(q).
\]

\paragraph{Direct monotonicity form.}
Let
\[
\bar\Delta_\mu(q)
=
\bar Q^1_\mu(q)-\bar Q^0_\mu(q),
\qquad
\Delta_\lambda(q,o)
=
Q^1_\lambda(q,o)-Q^0_\lambda(q,o).
\]
Then
\[
\Delta_\lambda(q,o)
=
\epsilon_o\bar\Delta_{\lambda/\epsilon_o}(q).
\]
Therefore, if $\bar\Delta_\mu(q)$ is non-increasing in $\mu$, then
$\Delta_\lambda(q,o)$ is non-increasing in $\lambda$. This gives the
passive-set monotonicity directly, without enumerating the two-dimensional
passive set.


\subsection{Refined Thompson-Learning Strategies}
\label{app:refinements}

\subsubsection{Adaptive TW}
\label{app:adaptive_tw}

The Adaptive-TW policy combines the Thompson-sampled Whittle
index with an uncertainty-aware UCB score. The visit-dependent
trust coefficient is

\begin{equation}
\tau_i(s,t)
=
\operatorname{clip}\!\left(
\frac{O_i(s,t)}
{O_i(s,t)+2+\sqrt{|\mathcal{S}|}},
\tau_{\min},
\tau_{\max}
\right),
\label{eq:app_trust}
\end{equation}

where $O_i(s,t)$ denotes the number of observed transitions
from state $s$ for arm $i$.

The resulting arm-selection score is

\begin{equation}
A_i(s,t)
=
\tau_i(s,t)\,
z\!\left(\widetilde W_i(s,t)\right)
+
\left[1-\tau_i(s,t)\right]\,
z\!\left(U_i(s,t)\right),
\label{eq:app_adaptive}
\end{equation}

where $z(\cdot)$ standardizes scores across arms.

\subsubsection{Offline (prior) refinement.}
The offline refinement constructs a structured transition prior
$P_{0,i}^{a}(s,\cdot)$ for each data center $i$ from historical job
traces (per-job core-hours, measured energy, and VM category) together
with the known cyclic structure of the job queue. For every queue
state $s$, we roll the queue forward over a short lookahead horizon of
$B$ batches and enumerate a small set of \emph{dispatch paths}, each
path being a way of partitioning the upcoming $B\cdot(\text{batch
size})$ jobs into a batch to run now and batches to defer. For each
path we solve a small combinatorial sub-problem that selects the
run-now batch so as to maximize the immediate locational energy-cost
saving relative to the default (in-order) batch, i.e.
$\text{LMP}(s,i),\big(\textstyle\sum_{j\in\text{default}}p_j -
\sum_{j\in\text{selected}}p_j\big)$, where $p_j$ is job $j$'s power. Figure \ref{offline_prior_under_different_paths} shows the scheduling selection under three scenarios. 
Each resulting candidate transition is assigned a non-negative weight
(uniform, power reduction, or the LMP-weighted saving above), and
$P_{0,i}^{a}(s,\cdot)$ is obtained by normalizing these weights over the
candidate set. The whole prior is built from logged data and the
queue's cyclic indexing, with no online interaction.

\subsubsection{Support refinement.}
The support refinement additionally defines a feasible-transition
graph. A dispatch path is admissible only if every batch along it
respects the per-batch core-hour capacity of the data center; paths
that violate this scheduling constraint are given zero weight and are
removed before normalization, so they carry no mass in
$P_{0,i}^{a}(s,\cdot)$ and no posterior probability can leak toward
them (Figure \ref{offline_prior_under_different_paths} (d) shows the infeasible path due to the core hour violation of computing nodes). When no admissible positive-weight candidate exists at a state,
the prior falls back to the least-power feasible grouping. Restricting
the support to transitions the scheduler can actually realize sharpens
the prior and keeps posterior updates consistent with the underlying
cyclic scheduling process.

\begin{figure*}[!h]
    \centering \includegraphics[width=0.7\linewidth]{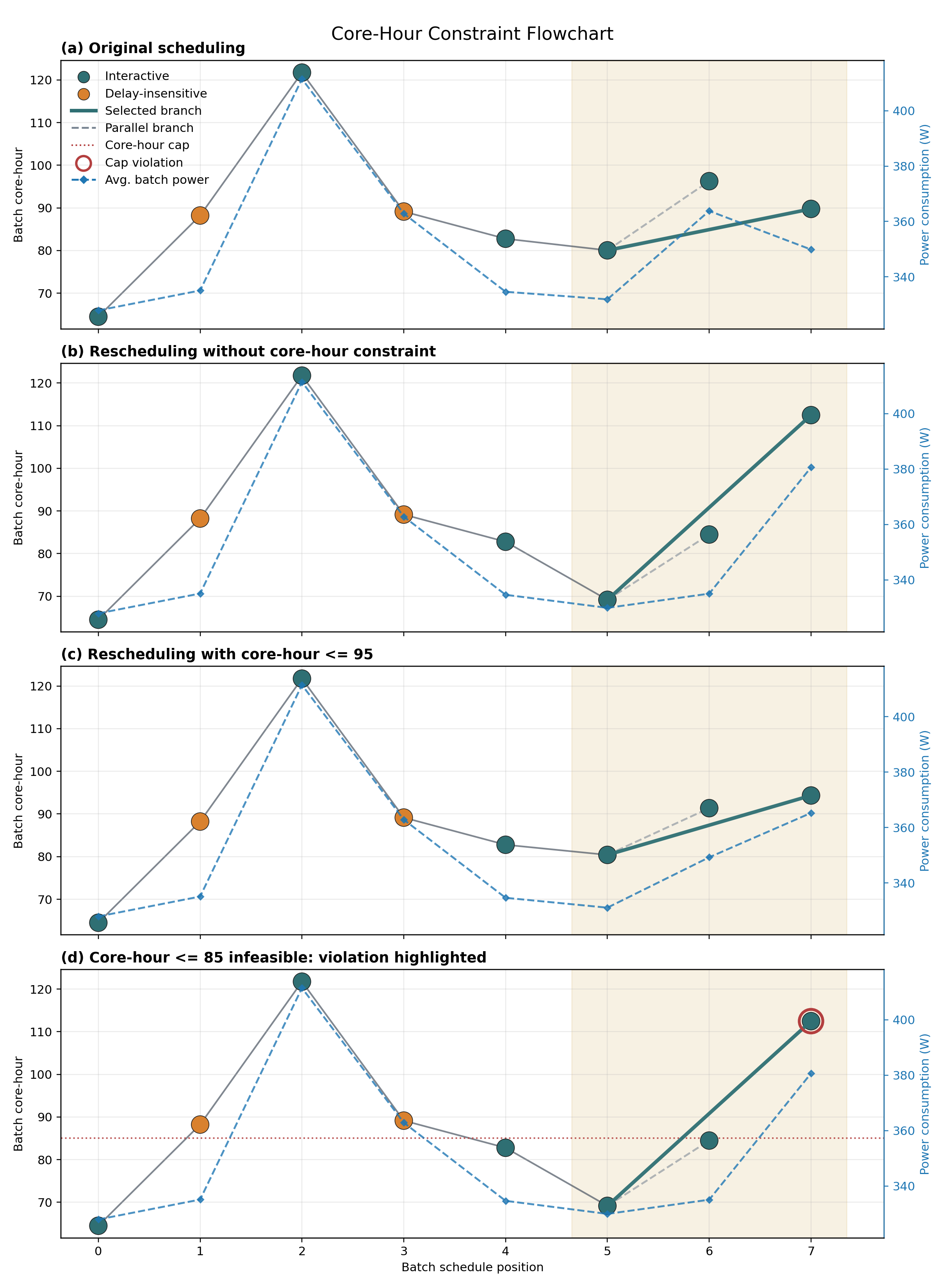}
    \caption{Job rescheduling under different scenarios for graph-based offline prior construction and updates: (a) Original scheduling, (b) Rescheduling without core-hour constraints, (c) Rescheduling with core-hours per batch less than 95, and  (d) Rescheduling with core-hours per batch less than 85, which is infeasible.}
    \label{offline_prior_under_different_paths}
\end{figure*}

Figure \ref{bandit_stragies_reward_compare} shows differences in rewards with and without offline support refinements when using five varying bandit strategies: EXP4, Local + global+ TW, Oracle Whittle, State Thompson, and TW only. With the weighted-updated offline refinement, all the Thompson-learning-based bandit strategies improve, while the Orcale whittle remains the same reward levels due to the pure Whittle index look-up.

\begin{figure*}[!h]
    \centering
    \includegraphics[width=0.6\linewidth]{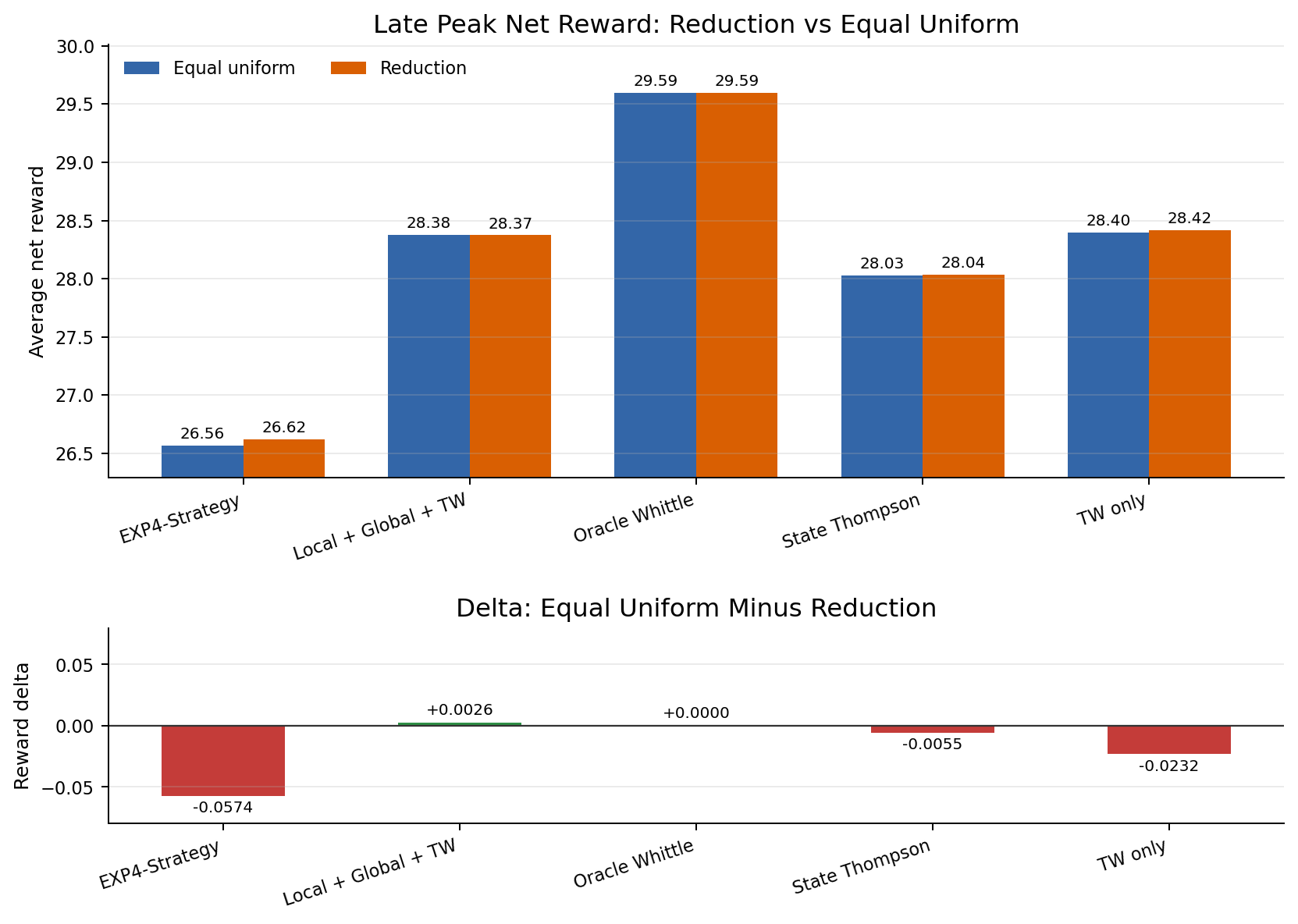}
    \caption{The differences in rewards with and without offline support refinements in five varying bandit strategies: EXP4, Local + global+ TW, Oracle Whittle, State Thompson, and TW only}
    \label{bandit_stragies_reward_compare}
\end{figure*}

\subsubsection{Gated Structural Prior}
\label{app:gated_prior}

Let $\widehat P_i^a(s,\cdot)$ denote a transition row sampled
from the Dirichlet--Categorical posterior and
$P_{0,i}^a(s,\cdot)$ denote the structural prior.
The transition model used for Whittle-index computation is

\begin{equation}
\widetilde P_i^a(s,\cdot)
=
\left[1-\gamma_i^a(s,t)\right]
\widehat P_i^a(s,\cdot)
+
\gamma_i^a(s,t)
P_{0,i}^a(s,\cdot),
\label{eq:app_gated_transition}
\end{equation}

where

\begin{equation}
\gamma_i^a(s,t)
=
\frac{2+\sqrt{|\mathcal{S}|}}
{2+\sqrt{|\mathcal{S}|}+O_i^a(s,t)}.
\label{eq:app_gate}
\end{equation}

Thus, the structural prior receives greater weight for
weakly visited state--action pairs and gradually vanishes as
transition observations accumulate.

\subsubsection{Beta-Gated Prior}
\label{app:beta_gate}

To reduce excessive conservativeness under sparse or noisy
observations, we also consider a stochastic gate:

\begin{equation}
G_i^a(s,t)
\sim
\operatorname{Beta}\!\left(
\max\{\kappa\gamma_i^a(s,t),\epsilon\},
\max\{\kappa[1-\gamma_i^a(s,t)],\epsilon\}
\right),
\label{eq:app_beta_gate}
\end{equation}

where $\kappa=20$ and $\epsilon=10^{-3}$.
The deterministic gate $\gamma_i^a(s,t)$ in
Eq.~\eqref{eq:app_gated_transition} is then replaced by
$G_i^a(s,t)$.


\subsection{Experimental Details}
\label{app:experiments}

\subsubsection{Datasets and Workload Construction}
\label{app:data}

We use the Microsoft Azure VM dataset reported in ~\citep{cortez_resource_2017} and the calibrated MIT SuperCloud dataset ~\citep{the_mit_supercloud_mit_2021} as representative datasets for the three experiments described above. Table \ref{tab:dataset-comparison} summarizes the key parameters of two datasets. The Microsoft Azure VM dataset contains millions of lightweight workloads, primarily corresponding to Microsoft internal services and testing jobs. However, it does not include information on the underlying physical infrastructure or power consumption. We therefore use the utilization and hardware assumptions to estimate power consumption.  

\begin{table*}[!h]
\centering
\small
\begin{tabular}{p{0.2\textwidth}p{0.35\textwidth}p{0.35\textwidth}}
\hline
\textbf{Aspects} & \textbf{Microsoft VM Datasets} & \textbf{MIT Supercloud Datasets} \\
\hline

Labels
& CPU utilization, lifetime, core-hour, work types (interactive and delay-insensitive), Parties 
& GPU/HPC jobs with model labels, CPU utilization, GPU power, and energy \\

Work types
& Light-weight VM jobs
& DNN/HPC workload such as BERT, Inception, ResNet, VGG, and unlabeled traces \\

Power consumptions
& $4\times10^{-4}$ W to 52.67 W (Synthesis)
& 266.9 W to 6.3 kW (Measurements)\\

Reference &~\cite{cortez_resource_2017} 
&~\cite{the_mit_supercloud_mit_2021}\\
\hline

\end{tabular}
\caption{Comparison of the Microsoft VM and MIT Supercloud datasets}
\label{tab:dataset-comparison}
\end{table*}

\begin{table*}[!h]
\centering
\small
\begin{tabular}{p{2cm} p{3.5cm} p{3.5cm} p{3.5cm}}
\toprule
\textbf{Model family} &
\textbf{Variants} &
\textbf{AI workloads} &
\textbf{Reasons for inclusion} \\
\midrule

\textbf{BERT} &
\makecell[l]{bert-base-uncased\\
distilbert-base-uncased} &
\makecell{Natural language \\ processing} &
Transformer-based NLP workloads dominated by matrix multiplication, attention, reduction, and elementwise operations. \\

\addlinespace

\textbf{ResNet} &
\makecell[l]{resnet50, resnet50\_v1.5\\
resnet101, resnet101\_v2\\
resnet152, resnet152\_v2} &
\makecell{Computer vision \\ (classification)} &
Provides multiple network depths and architectural revisions within a widely used CNN family. \\

\addlinespace

\textbf{VGG} &
\makecell[l]{vgg11\\
vgg16\\
vgg19} &
\makecell{Computer vision \\ (classification)} &
Represents simple sequential convolutional networks with increasing depth. \\

\addlinespace

\textbf{U-Net} &
\makecell[l]{U3-32, U3-64, U3-128\\
U4-32, U4-64, U4-128\\
U5-32, U5-64, U5-128} &
\makecell{Computer vision\\ (segmentation)} &
Represents encoder--decoder segmentation workloads with multiple architectural and model-size configurations. \\

\addlinespace

\textbf{Inception} &
\makecell[l]{inception3\\
inception4} &
\makecell{Computer vision \\ (classification)} &
Represents multi-branch CNN architectures operating at multiple spatial scales. \\

\bottomrule
\end{tabular}
\caption{Representative AI workloads from the MIT Supercloud datasets.}
\label{tab:mit_supercloud_models}
\end{table*}

The MIT Supercloud dataset is a large-scale monitoring dataset that contains a sample of labeled AI workload traces, scheduler logs, system metrics, and high-frequency GPU monitoring data. GPU measurements were collected at 100-millisecond intervals and include GPU utilization, memory utilization, and power drawn, providing detailed measurements of workload-level power usage and energy consumption. AI workloads include labeled workloads in different model families and unlabeled workloads. Table \ref{tab:mit_supercloud_models} records representative AI workloads from the MIT Supercloud datasets in five model families. We use the labeled workloads sampled evenly from five model families. 

\subsubsection{Contextual feature construction}

From two datasets, we extracted the following features for each job batch for bandit learning strategies: (a) average power consumption, (b) CPU/GPU utilization, (b) average normalized core-hour, (c) fraction of delay-sensitive jobs, and (d) normalized state scaling between 0 and 1, which shows the position in the job queue. For MIT Supercloud, we manually assume the interactive or delay-insensitive job fractions. Figure \ref{fig:contextual_metrics} shows the contextual metrics used for MIT GPU workloads: average CPU/GPU utilization, core-hours, average power (W), measured energy of work traces, power-saving percentage, and interactive share. The QoS cost is calculated based on the interactive share of the job batch.

\begin{figure*}[!h]
    \centering
    \includegraphics[width=\linewidth]{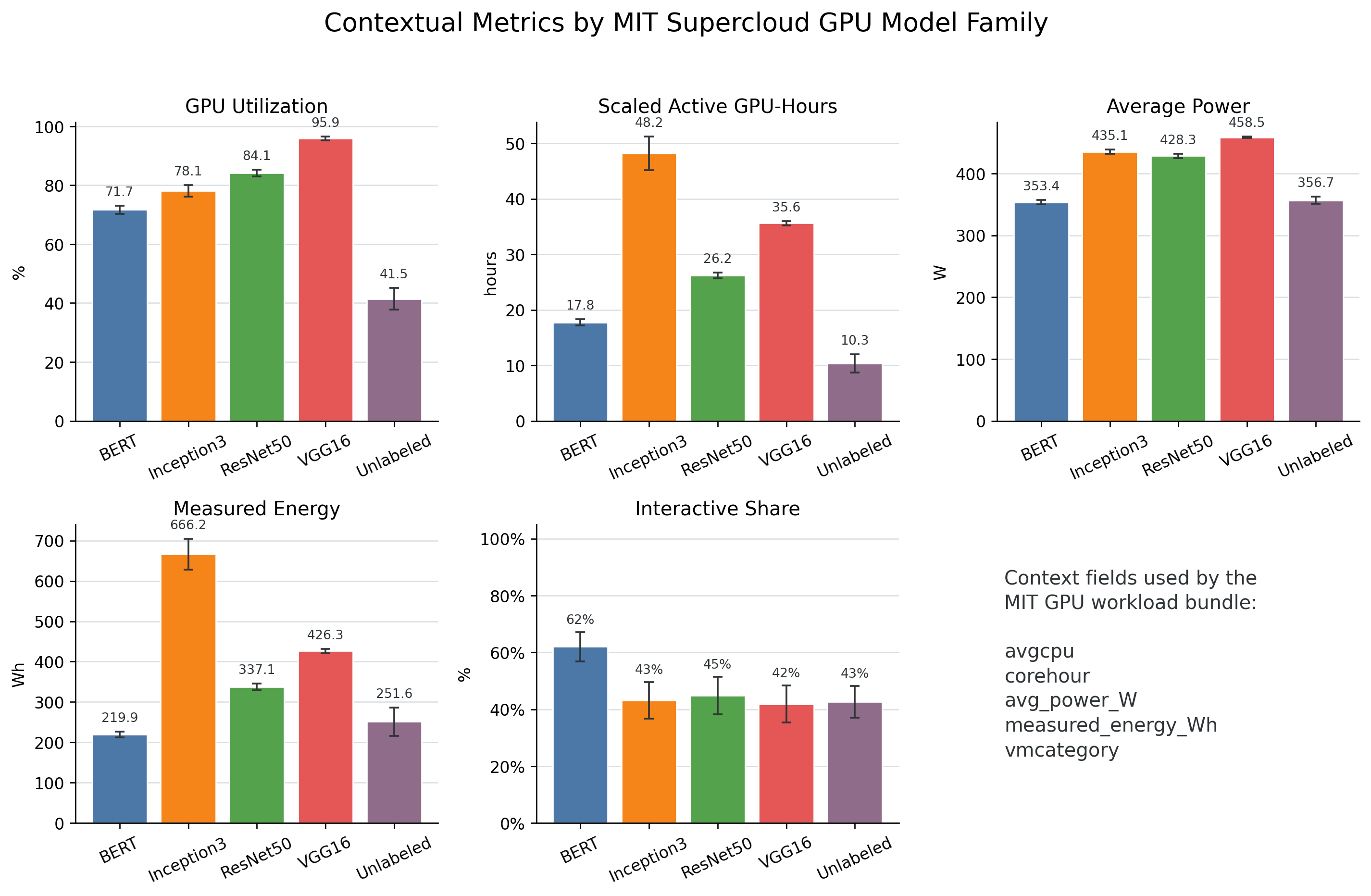}
    \caption{Contextual metrics used by MIT GPU workloads: Average CPU/GPU utilization, core-hours, average power (W), measured energy of work traces, power saving index (percentage)}
    \label{fig:contextual_metrics}
\end{figure*}

\section{Additional Ablation Results}
\label{app:ablation}

\begin{table*}[!h]
\centering
\caption{Paired per-seed differences in cumulative average reward (percentage points of Oracle Whittle) between refined TW variants and a reference policy, over the same 10 seeds. Entries are mean difference $\pm$ 95\% confidence half-width (seeds where the refined variant wins). $^{*}$: paired $t$-test $p<0.05$; bold: significant after Holm correction within each reference.}
\label{tab:baseline_setting_paired_significance}
\begin{tabular}{lcccc}
\toprule
 & \multicolumn{2}{c}{vs.\ TW} & \multicolumn{2}{c}{vs.\ EXP4} \\
\cmidrule(lr){2-3}\cmidrule(lr){4-5}
Method & Round 100 & Round 1000 & Round 100 & Round 1000 \\
\midrule
Adaptive TW & +9.6 $\pm$ 8.6 (8/10)$^{*}$ & +2.7 $\pm$ 1.7 (8/10)$^{*}$ & +1.0 $\pm$ 4.9 (6/10) & +2.3 $\pm$ 2.8 (7/10) \\
Offline prior & +8.8 $\pm$ 7.5 (8/10)$^{*}$ & +2.2 $\pm$ 1.3 (9/10)$^{*}$ & +0.2 $\pm$ 4.7 (6/10) & +1.8 $\pm$ 2.2 (6/10) \\
Gated prior & +7.9 $\pm$ 8.8 (7/10) & +1.8 $\pm$ 3.1 (8/10) & $-$0.8 $\pm$ 6.4 (5/10) & +1.4 $\pm$ 2.7 (5/10) \\
Beta-gate prior & +6.7 $\pm$ 6.4 (8/10)$^{*}$ & +1.1 $\pm$ 1.6 (7/10) & $-$1.9 $\pm$ 5.2 (4/10) & +0.7 $\pm$ 2.6 (6/10) \\
Support + offline & +10.9 $\pm$ 6.9 (9/10)$^{*}$ & +2.2 $\pm$ 1.7 (8/10)$^{*}$ & +2.2 $\pm$ 5.8 (5/10) & +1.8 $\pm$ 3.1 (5/10) \\
\bottomrule
\end{tabular}
\end{table*}

\subsection{L1 error and Off-Support Leakage}
\label{app:leakage}

\begin{table*}[!h]
\centering
\caption{Baseline-setting checkpoint diagnostics ranked by reward versus oracle. Rank 1 is best within each checkpoint round.}
\label{tab:baseline_setting_checkpoint_diagnostics_ranked}
\resizebox{0.7\linewidth}{!}{%
\begin{tabular}{rlrrrrr}
\toprule
Round & Method & Reward vs. oracle (\%) $\downarrow$ & L1 error & Leakage & Top-1 & Top-2 \\
\midrule
100 & Support + offline & 94.82 & 0.078 & 0.000 & 0.631 & 0.815 \\
100 & Offline prior & 93.83 & 0.108 & 0.026 & 0.570 & 0.771 \\
100 & Adaptive TW & 93.35 & 0.182 & 0.043 & 0.572 & 0.785 \\
100 & Gated prior & 92.09 & 0.109 & 0.026 & 0.598 & 0.775 \\
100 & Beta-gate prior & 91.64 & 0.110 & 0.026 & 0.573 & 0.744 \\
100 & EXP4 & 90.36 & 0.179 & 0.044 & 0.580 & 0.758 \\
100 & ST & 89.91 & 0.170 & 0.043 & 0.530 & 0.702 \\
100 & Local UCB+TW & 89.48 & 0.176 & 0.039 & 0.558 & 0.751 \\
100 & Global UCB+TW & 86.67 & 0.179 & 0.040 & 0.511 & 0.710 \\
100 & TW & 84.14 & 0.185 & 0.036 & 0.514 & 0.691 \\
\addlinespace
1000 & Adaptive TW & 97.76 & 0.111 & 0.019 & 0.779 & 0.912 \\
1000 & Beta-gate prior & 97.41 & 0.069 & 0.014 & 0.786 & 0.923 \\
1000 & EXP4 & 97.16 & 0.117 & 0.033 & 0.647 & 0.825 \\
1000 & ST & 96.91 & 0.107 & 0.027 & 0.633 & 0.826 \\
1000 & Gated prior & 96.74 & 0.070 & 0.014 & 0.745 & 0.891 \\
1000 & Support + offline & 96.66 & 0.055 & 0.000 & 0.765 & 0.912 \\
1000 & Offline prior & 96.53 & 0.070 & 0.013 & 0.757 & 0.906 \\
1000 & Local UCB+TW & 95.42 & 0.111 & 0.018 & 0.751 & 0.899 \\
1000 & TW & 94.68 & 0.115 & 0.018 & 0.735 & 0.897 \\
1000 & Global UCB+TW & 94.29 & 0.109 & 0.018 & 0.721 & 0.881 \\
\bottomrule
\end{tabular}%
}
\end{table*}

\begin{table*}[!h]
\centering
\caption{Baseline-setting diagnostic reductions from round 100 to round 1000. Rows are ordered by the round-1000 reward-versus-oracle ranking.}
\label{tab:baseline_setting_checkpoint_reductions}
\resizebox{0.7\linewidth}{!}{\begin{tabular}{lrrrrr}
\toprule
Strategy & Reward rank & $\Delta$ L1 & L1 red. (\%) & $\Delta$ leakage & Leakage red. (\%) \\
\midrule
Adaptive TW & 1 & 0.07117 & 39.1 & 0.023919 & 55.4 \\
Beta-gate prior & 2 & 0.04087 & 37.1 & 0.012222 & 46.7 \\
EXP4 & 3 & 0.06119 & 34.2 & 0.010547 & 24.1 \\
ST & 4 & 0.06320 & 37.1 & 0.015615 & 36.7 \\
Gated prior & 5 & 0.03854 & 35.5 & 0.011805 & 45.5 \\
Support + offline & 6 & 0.02308 & 29.4 & 1.0392e-07 & 46.1 \\
Offline prior & 7 & 0.03784 & 35.2 & 0.013506 & 51.4 \\
Local UCB+TW & 8 & 0.06449 & 36.7 & 0.020952 & 53.7 \\
TW & 9 & 0.07014 & 38.0 & 0.018581 & 51.0 \\
Global UCB+TW & 10 & 0.07067 & 39.4 & 0.021871 & 54.4 \\
\bottomrule
\end{tabular}}
\end{table*}


\subsection{Contextual-Noise Robustness Tests}
\label{app:noise}

We evaluate robustness under sparse observations and corrupted
contextual states. At each round $t$, the true state $s_{i,t}$
is replaced by a uniformly sampled state with probability
$\rho$; otherwise, it is observed correctly. We consider

\[
|\mathcal{S}|
\in
\{8,20,50,100\},
\qquad
\rho
\in
\{0,0.1,0.2,0.3\},
\]

resulting in 16 stress-test settings. Each experiment uses
100 learning rounds. Performance is reported as cumulative
reward normalized by the Oracle Whittle benchmark.

\begin{table*}[!h]
\centering
\caption{Robustness to contextual noise. For each $(|S|,\rho)$ setting, we report the best baseline, TW, and the best refined TW variant. Margins compare the best refined variant against the baseline and TW.}
\label{tab:sweep_a_real_best_beta_gate}
\resizebox{\linewidth}{!}{%
\begin{tabular}{cclccclcc}
\toprule
$|S|$ & Noise & Best baseline name & Best base & TM--TW & Best refined & Best refined variant & Margin vs baseline & Margin vs TW \\
\midrule
8 & 0.0 & EXP4 & 88.27 & 80.53 & 95.50 & \textbf{Adaptive TM--TW + support/offline prior} & +7.22$\pm$7.86$^{*}$ & +14.54$\pm$9.62$^{**}$ \\
8 & 0.1 & EXP4 & 89.83 & 82.88 & 90.86 & \textbf{Adaptive TM--TW} & +1.03$\pm$9.01 & +17.83$\pm$12.78$^{**}$ \\
8 & 0.2 & EXP4 & 78.45 & 74.67 & 87.74 & \textbf{Adaptive TM--TW} & +9.30$\pm$17.58 & +12.65$\pm$12.22$^{**}$ \\
8 & 0.3 & EXP4 & 86.09 & 71.79 & 87.77 & \textbf{Adaptive TM--TW + gated prior} & +1.67$\pm$8.23 & +18.78$\pm$14.60$^{**}$ \\
\addlinespace
20 & 0.0 & EXP4 & 79.09 & 63.20 & 86.64 & \textbf{Adaptive TM--TW + support/offline prior} & +7.55$\pm$12.31 & +21.86$\pm$14.81$^{**}$ \\
20 & 0.1 & EXP4 & 75.17 & 59.86 & 81.40 & \textbf{Adaptive TM--TW} & +6.22$\pm$11.83 & +24.95$\pm$13.18$^{***}$ \\
20 & 0.2 & EXP4 & 74.47 & 63.39 & 78.67 & \textbf{Adaptive TM--TW} & +4.20$\pm$9.20 & +22.08$\pm$17.05$^{**}$ \\
20 & 0.3 & EXP4 & 70.19 & 54.79 & 77.26 & \textbf{Adaptive TM--TW + gated prior} & +7.07$\pm$15.03 & +25.64$\pm$15.25$^{***}$ \\
\addlinespace
50 & 0.0 & EXP4 & 79.97 & 58.09 & 82.04 & \textbf{Adaptive TM--TW + offline prior} & +2.07$\pm$9.70 & +32.88$\pm$13.45$^{***}$ \\
50 & 0.1 & EXP4 & 72.81 & 58.54 & 77.55 & \textbf{Adaptive TM--TW + offline prior} & +4.74$\pm$14.36 & +25.05$\pm$10.79$^{***}$ \\
50 & 0.2 & EXP4 & 74.44 & 58.10 & 77.03 & \textbf{Adaptive TM--TW} & +2.59$\pm$11.23 & +22.72$\pm$11.85$^{***}$ \\
50 & 0.3 & EXP4 & 64.98 & 54.25 & 73.78 & \textbf{Adaptive TM--TW + gated prior} & +8.79$\pm$13.86 & +22.32$\pm$9.83$^{***}$ \\
\addlinespace
100 & 0.0 & EXP4 & 72.40 & 55.50 & 79.21 & \textbf{Adaptive TM--TW + gated prior} & +6.81$\pm$8.06$^{*}$ & +31.83$\pm$15.06$^{***}$ \\
100 & 0.1 & EXP4 & 71.28 & 53.51 & 77.68 & \textbf{Adaptive TM--TW + gated prior} & +6.40$\pm$7.83$^{*}$ & +34.02$\pm$10.07$^{***}$ \\
100 & 0.2 & EXP4 & 73.66 & 43.67 & 72.84 & \textbf{Adaptive TM--TW + support/offline prior} & -0.82$\pm$12.48 & +30.24$\pm$14.65$^{***}$ \\
100 & 0.3 & EXP4 & 63.94 & 51.39 & 71.08 & \textbf{Adaptive TM--TW + gated prior} & +7.14$\pm$11.97 & +21.28$\pm$12.32$^{***}$ \\
\bottomrule
\end{tabular}%
}
\vspace{0.25em}
\footnotesize{Margins are mean$\pm$one standard deviation over 10 paired seeds; $^{*}p<0.05$, $^{**}p<0.01$, $^{***}p<0.001$ by paired $t$-test.}
\end{table*}

\subsection{Real-World Evaluation with AI Workloads and Electricity Prices}
\label{Real_world}

Electricity prices are taken from \citep{noauthor_electricity_2026} eight data centers in Texas area, which is presented in Figure. \ref{fig:electricity_prices}. The ERCOT LMP analysis was separated into two parts, in which four of operating hyperscale data centers are chosen from the Northwestern region with low LMPs due to high renewable deployment and four of small-scale operating data centers were chosen from the Urban South region. We used the electricity prices in the peak hour (i.e., 8pm to 9pm).

\begin{figure*}[!h]
    \centering
    \includegraphics[width=0.8\linewidth]{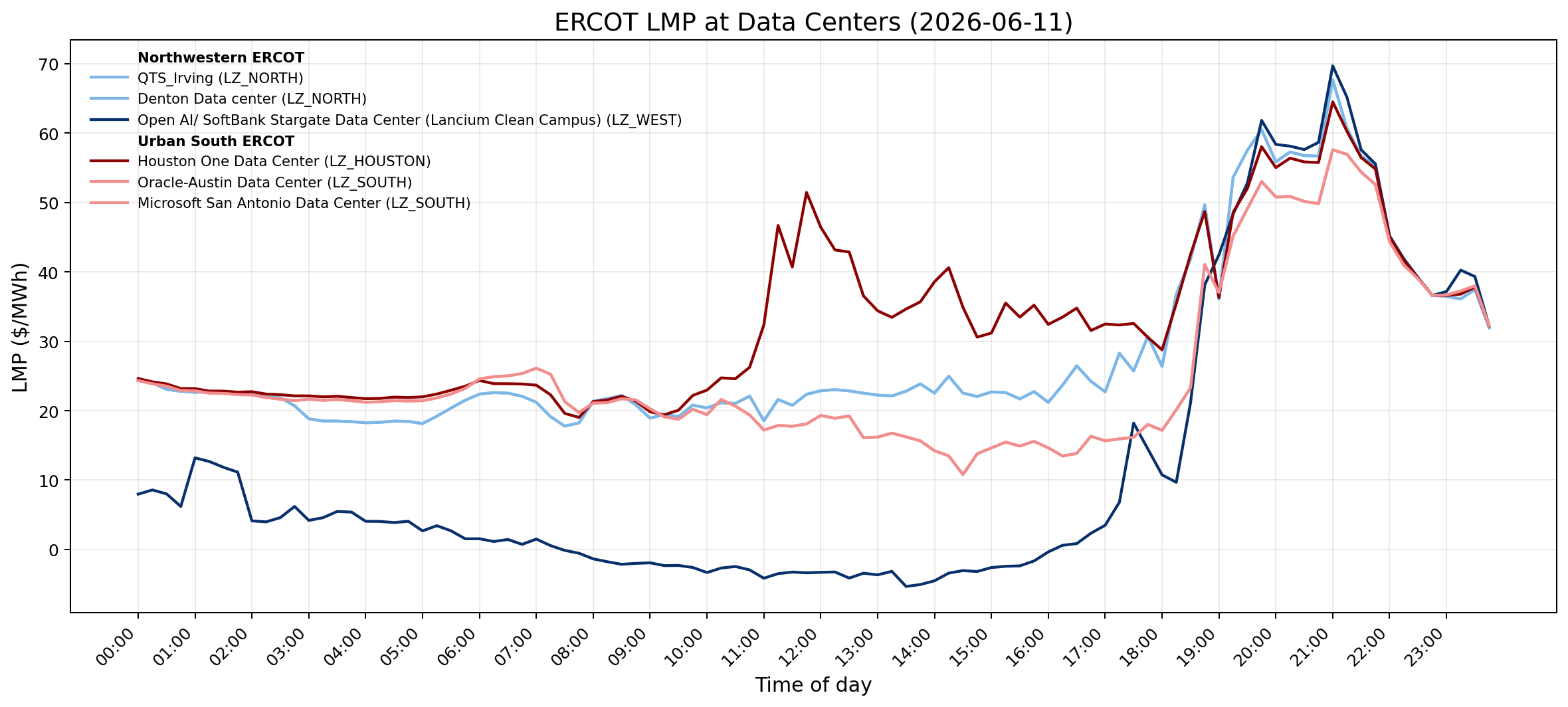}
    \caption{ERCOT locational marginal prices (LMPs) at eight selected data centers. LMPs are divided into two families: Northwestern ERCOT and Urban South ERCOT}
    \label{fig:electricity_prices}
\end{figure*}



\subsubsection{Implementation and Hyperparameters}
\label{best_beta}

We sweep the warm-up rounds for adaptive trusts to ensure the best improvement in refinement strategies. Table \ref{tab:warmup-sweep} shows the sweep results and the optimal warm-up rounds of 10 is used for experiments. 

\begin{table*}[!h]
\centering
\caption{Warm-up sweep for the late-peak refinement waterfall.
The adaptive Local+Global+TW mixture only departs from the equal-weight TW policy at a
10-round warm-up; longer warm-ups pin it to TW, making the trust ceiling $\rho_{\max}$
irrelevant.}
\label{tab:warmup-sweep}
\resizebox{0.8\linewidth}{!}{%
\begin{tabular}{ccccccc}
\toprule
Warm-up & $\rho_{\max}$ & TW only & Adaptive 
& $\Delta$ (Adap.\,$-$\,TW) & \% Oracle & Win rate \\
(rounds) & & (\$/MWh) & (\$/MWh) & (\$/MWh) & & \\
\midrule
\textbf{10} & \textbf{0.98} & \textbf{30.604} & \textbf{31.774} 
& \textbf{+1.170} & \textbf{99.42\%} & \textbf{0.50} \\
20 & 0.98 & 31.620 & 31.620 & 0.000 & 98.94\% & 0.00 \\
40 & 0.98 & 31.291 & 31.291 & 0.000 & 97.91\% & 0.00 \\
40 & 0.95 & 31.291 & 31.291 & 0.000 & 97.91\% & 0.00 \\
40 & 0.90 & 31.291 & 31.291 & 0.000 & 97.91\% & 0.00 \\
\bottomrule
\end{tabular}%
}
\end{table*}

\newpage

\subsection{Additional Discussion}
\label{app:discussion}

\subsubsection{Computational Considerations}

All experiments are conducted on an Apple M1 Pro with 14 GB of memory. For the largest state space considered, $|S|=100$, a 600-round bandit simulation completes within 10 minutes, corresponding to approximately one second per round. This computational efficiency demonstrates the practical feasibility of the proposed approach for real-time learning and communication.

\subsubsection{Limitations}

Effective deployment requires dataset-specific calibration, graph-based offline prior construction, and careful tuning of refinement hyperparameters for different data center environments.



\end{document}